\documentclass[10pt,oneside, reqno] {amsart}
\usepackage{hyperref}
\usepackage{latexsym,amsmath,amsfonts,amscd,amssymb}
\usepackage{mathtools}
\usepackage{tikz-cd}
\usepackage{pdfpages}
\usepackage{graphicx}
\usepackage{marginnote}
\usepackage{amssymb}
\usepackage{mathabx}
\usepackage[margin=3.0cm]{geometry}
\usepackage{multirow}
\usepackage{array}
\usepackage{graphicx}
\usepackage[draft]{fixme}
\usepackage{comment}

\usepackage{enumerate}
\theoremstyle{plain}  
\newtheorem{theorem}{Theorem}[section]
\newtheorem*{theorem*}{Theorem}
\newtheorem{corollary}[theorem]{Corollary}
\newtheorem{lemma}[theorem]{Lemma}
\newtheorem{proposition}[theorem]{Proposition}

\theoremstyle{remark}
\newtheorem{example}[theorem]{Example}

\newtheorem*{notation*}{Notation}
\newtheorem{remark}[theorem]{Remark}

\newtheorem*{question*}{Question}

\newtheorem*{claim*}{Claim}

\newcommand{\li}{\operatorname{li}}

\newcommand{\norm}[1]{\lVert#1\rVert}

\renewcommand{\leq}{\leqslant}
\renewcommand{\le}{\leqslant}

\renewcommand{\ge}{\geqslant}

\newcommand{\Z}{\mathbb{Z}}
\newcommand{\C}{\mathbb{C}}

\newcommand{\HH}{\mathbb{H}}
\newcommand{\PP}{\mathbb{P}}

\newcommand{\SU}{\mathrm{SU}}

\newcommand{\GL}{\mathrm{GL}}
\newcommand{\SL}{\mathrm{SL}}

\newcommand{\Sp}{\mathrm{Sp}}

\newcommand{\GCD}{\mathrm{GCD}}

\newcommand{\pkp}{\raise1pt\hbox{\ensuremath{\scriptscriptstyle(}}k\raise1pt\hbox{\ensuremath{\scriptscriptstyle)}}}
\newcommand{\pip}{\raise1pt\hbox{\ensuremath{\scriptscriptstyle(}}i\raise1pt\hbox{\ensuremath{\scriptscriptstyle)}}}

\let\lowchi\chi
\renewcommand{\chi}{\mathchoice{\raise1.5pt\hbox{\ensuremath{\lowchi}}}{\raise1.5pt\hbox{\ensuremath{\lowchi}}}{\lowchi}{\lowchi}}

\DeclareMathOperator{\tr}{tr}

\DeclareMathOperator{\Hom}{Hom}

\DeclareMathOperator{\Id}{Id}

\newcommand{\Aut}{\operatorname{Aut}}

\newcommand{\FF}{\mathbb{F}}

\newcommand{\QQ}{\mathbb{Q}}
\newcommand{\RR}{\mathbb{R}}

\newcommand{\ZZ}{\mathbb{Z}}

\newcommand{\Gal}{\text{Gal}}

\newcommand{\Tr}{\operatorname{Tr}}
\newcommand{\Det}{\operatorname{Det}}

\begin{document}
\title{Remarks on the $\ZZ/p$ Dijkgraaf-Witten invariants of 3D mapping tori}
\author{William  Chen}
\thanks{}
\address{Department of Mathematics, Rutgers University}
\email{oxeimon@gmail.com}
\author{Alex Kontorovich}
\thanks{A. Kontorovich supported under NSF grants: DMS-1802119 and BSF 2020119.}
\address{Department of Mathematics, Rutgers University}
\email{alex.kontorovich@rutgers.edu}
\author{Shehryar Sikander}
\thanks{S. Sikander supported under  US DOE grant: DOE-SC0010008 to Rutgers}
\address{New High Energy Theory Center, Rutgers University}
\email{shehryar.sikander@physics.rutgers.edu}

\begin{abstract}
We make some remarks on the $\Z/p$ Dijkgraaf-Witten invariants of 3D mapping tori and determine the asymptotic behavior of their sum over all diffeomorphism classes of 3D mapping tori of genus one.
\end{abstract}

\maketitle
\tableofcontents

\section{Introduction}
Starting with the Euclidean gravitational path integral defined by Hartle and Hawking, sums over diffeomorphism classes of D dimensional manifolds are by now ubiquitous  in quantum gravity. Recently, Banerjee and Moore \cite{BM} considered sums of TQFT invariants over equivalence classes of one and two dimensional cobordisms. In this note we study sums of $\ZZ/p$ Dijkgraaf-Witten invariants  over diffeomorphism classes of 3D mapping tori.

Let $G$ be a finite group. Using gauge theory with gauge group $G$, Dijkgraaf and Witten \cite{DW} constructed invariants of oriented three manifolds which depend on a choice of a class in $H^3(BG, U(1))$.  Let $S$ be a closed, connected, and oriented surface of genus bigger than zero, let $\Gamma(S)$ denote its mapping class group, and let $[\Gamma(S)]$ be the set of conjugacy classes of $\Gamma(S)$. For $f$ in $ \Gamma(S)$ let $M(f)$ denote the three dimensional mapping torus of $f$. Since conjugate mapping classes give  diffeomorphic mapping tori, the Dijkgraaf-Witten invariant only depends on  the conjugacy class $[f]$ in $[\Gamma(S)]$. Let $Z(M([f]), G, \omega)$ denote the  Dijkgraaf-Witten invariant for some $\omega$ in $H^3(BG, U(1))$, and let $H : [\Gamma(S)] \to \RR$ be a suitable ordering function. We seek the asymptotics of the sum 
\begin{align}
   \sum_{[f] \in [\Gamma(S)],\atop H([f])<T} Z(M([f], G, \omega) 
\end{align}
as $T\to \infty$. Let $S$ be of genus one, then $\Gamma(S)$ is isomorphic to $\SL(2, \ZZ)$ and the absolute value of the trace provides a natural ordering function. Let $p$ be an odd prime and for a conjugacy class $[A]$ in $[\SL(2, \ZZ)]$ let $Z(M([A]), \ZZ/p)$ denote the $\ZZ/p$ Dijkgraaf-Witten invariant, with $\omega$ the identity class
in $H^3(B(\ZZ/p), U(1))$, of the mapping torus $M([A])$.  Using the formula for $Z(M([A], \ZZ/p)$ derived in Corollary \ref{cor} and Sarnak's Chebotarev density theorem for prime geodesics, we prove the following.
\begin{theorem}\label{thm:mainasymptotics}
Let $p$ be a fixed odd prime. Then the $\ZZ/p$ Dijkgraaf-Witten invariants of genus one mapping tori satisfy 
    \begin{align}
    \sum_{[A]\in[\SL(2,\Z)],\atop |Tr([A])|<T}  Z(M([A]), \ZZ/p) = \frac{2 p^3-2 p+1}{p^3-p} 
\li(T^2)
+
O(T^{3/2}(\log T)^2),
\end{align}    
as $T\to \infty$. Here
\begin{align}
    \li(x) = \int_2^{x} \frac{dt}{\log t},
\end{align}
is the logarithmic integral function (as in the prime number theorem).
\end{theorem}
 In the case $p=2$, similar asymptotics can be derived
using results of Section \ref{sec:2.3}. In Section \ref{sec:2.3}, we also prove that the numbers $Z(M([A]), \ZZ/2) +2$ are Fourier coefficients of a modular form of weight one, see Theorem \ref{cool1} for details. \\

Let $S$ be of arbitrary genus $g$ and for a mapping class $f$ in $\Gamma(S)$ denote by $\hat{f}$ the matrix in  $\Sp(2g, \ZZ)$ which gives the action of $f$ on $H_1(S, \ZZ)$, with respect to a fixed basis of $H_1(S, \ZZ)$. It is classical that the Smith normal form of the matrix $(\hat{f} - \Id)$ determines $H_1(M([f]), \ZZ)$, and hence the invariants $Z([f], \ZZ/p, \omega)$ when $\omega$ is the identity class, completely, see \eqref{hommf}, \eqref{coker}, and \eqref{ai}. The duality between the 3D Dijkgraaf-Witten theory and 2D rational conformal field theory (holomorphic orbifold  models in particular), allows one to use methods of rational conformal field theory to study Smith normal forms of matrices $(\hat{f} - \Id)$ where $\hat{f}$ is in the symplectic group. As an application of these 2D methods and Sarnak's Chebotarev density theorem, we prove the following refinements of Theorem \ref{thm:mainasymptotics}.

\begin{theorem}\label{thm:snfasymptotics}
For $[A]$ in $[\SL(2,\ZZ)]$, let $A(1)$ and $A(2)$ denote the diagonal entries of the Smith normal form of the matrix $(A-\Id)$, see \eqref{snf}. Let $p$ be an odd prime. Then 
    \begin{align}
         \sum_{[A]\in[\SL(2,\Z)], |Tr([A])|<T,\atop p \vert A(1)\, \text{and} \,p \vert A(2) }  &= \frac{1}{p^3-p}\li(T^2)
+
O(T^{3/2}(\log T)^2)\\
         \sum_{[A]\in[\SL(2,\Z)], |Tr([A])|<T,\atop p \nmid A(1)\, \text{and} \,p \vert A(2) }  &= \frac{p^2-1}{p^3-p}\li(T^2)
+
O(T^{3/2}(\log T)^2)\\
         \sum_{[A]\in[\SL(2,\Z)], |Tr([A])|<T,\atop p \nmid A(1)\, \text{and} \,p \nmid A(2) }  &= \frac{p^3-p^2-p-1}{p^3-p}\li(T^2)
+
O(T^{3/2}(\log T)^2)
    \end{align}
    as $T\to \infty$.     
\end{theorem}
These asymptotics should be compared with the nice results in \cite{SW}. For an application of these 2D methods in genus greater than one, see \eqref{2dsnf}. Perhaps these 2D methods can be leveraged to give faster than the state of the art algorithms for computing the Smith normal forms of  $(\hat{f} - \Id)$; one wonders if and how the topological quantum computers based on 2D rational conformal fields theories can calculate Smith normal forms faster than classical computers.

\subsection{Acknowledgements}Gregory Moore, Daniel Friedan, Henryk Iwaniec, Will Sawin, and  Anindya Banerjee are heartily thanked for discussions and encouragement.  

\section{Preliminaries: $\ZZ/p$ Dijkgraaf-Witten invariants of 3D mapping tori and the Smith normal form}

Let $G$ be a finite group and let $M$ be a closed, connected, and oriented three manifold. Dijkgraaf and Witten introduced an invariant of $M$, which depends on the choice of a class in $H^3(G, U(1))$, as follows: 
Any $\rho$ in 
$\Hom(\pi_1(M), G) $ determines a principal $G$-bundle  $E_{\rho}\to M$ as the quotient 
\begin{align}
    \tilde{M} \times_{\pi_1(M)} G := (\tilde{M}\times G)/\pi_1(M)
\end{align}
where $\tilde{M}$ is the universal cover of $M$ and $\pi_1(X)$ acts on $\tilde{M}$ as the deck group and acts on $G$ by left multiplication via $\rho$. The existence of the universal $G$-bundle $EG\to BG$ on the classifying space $BG$ of the group $G$ implies that there exists a unique (up to homotopy) continuous map 
\begin{align}
\label{hatrho1}
    \hat{\rho}: M \to BG 
\end{align}
such that 
\begin{align}
    E_{\rho} \cong \hat{\rho}^*(EG) 
\end{align}
 While isomorphism classes of $G$-bundles are parametrized by the set $\Hom(\pi_1(M), G)/G$, the set $\Hom(\pi_1(M), G)$ parametrizes homotopy classes of pointed maps; if one chooses a point in $M$ and a point in $BG$, then there is an equivalence between the set $\Hom(\pi_1(M), G)$, where $\pi_1(M)$ and $\pi_1(BG)\cong G$
 are defined with respect to the chosen point in $M$ and $BG$, and the set of homotopy equivalence classes of continuous maps from $M$ to $BG$ which map the chosen point in $M$ to the chosen point in $BG$. 

Let $\omega$ be a class in $H^3(BG, U(1))$ and let $\rho$ be a homomorphism in $\Hom(\pi_1(M), G) $. Using the continuous map $\hat{\rho}$ from \eqref{hatrho1} associated with $\rho$, one pulls back the class $\omega$ to obtain the class 
\begin{align}
\label{pullback} 
    \hat{\rho}^*(\omega)  \in H^3(M, U(1))
\end{align}
on the three manifold $M$. The pulled back class \eqref{pullback} only depends on the homotopy class of $\hat{\rho}$ and in this way one associates a well defined cohomology class in $H^3(M, U(1))$ with every $\rho$ in $\Hom(\pi_1(M), G)$.  The top degree  homology group $H_3(M, \ZZ)$  is infinite cyclic and the orientation of $M$ determines a generator of this group. Denote this generator by $[M]$ which is also called the fundamental class of $M$. One can now evaluate the pull-back \eqref{pullback} on the fundamental class, i.e.  
\begin{align}
    \langle\hat{\rho}^*(\omega) , [M] \rangle= \int_{M} \hat{\rho}^*(\omega) \in U(1)
\end{align}
The Dijkgraaf-Witten invariant of $M$ with respect to $\omega$ in $H^3(BG, U(1))$ is defined as 
\begin{align}
    \label{dwinv}
    Z(M, G, \omega) := \frac{1}{|G|} \sum_{\rho\in\Hom(\pi_1(M), G)}  \langle\hat{\rho}^*(\omega) , [M] \rangle
\end{align}
In case one chooses $\omega$ to be the identity element of the group $H^3(M, U(1))$, the formula \eqref{dwinv} specializes to 
\begin{align}
\label{dwuntwisted}
    Z(M, G) := \frac{|\Hom(\pi_1(M), G)|}{|G|} 
\end{align}
and will be refered to as the untwisted Dijkgraaf-Witten invariant. 

\begin{example}
    This example is taken from \cite{MT} and will be used later. Let $G=\ZZ/2$, its classifying space is the infinite dimensional real projective space
    \begin{align}
        BG= \RR\PP^{\infty}
    \end{align}
The authors  work with cohomology of $BG$ with $\ZZ/2$ coeffecients rather than with $U(1)$ coefficients. We have that 
\begin{displaymath}
H^3(\RR\PP^{\infty}, \ZZ/2) = \ZZ/2 
\end{displaymath} 
Let $\alpha$ denote the non-trivial element of $H^3(\RR\PP^{\infty}, \ZZ/2)$. For any $\rho$ in $\Hom(\pi_1(M), G)$ the evaluation 
\begin{align}
    \langle\hat{\rho}^*(\alpha), [M]\rangle \in \ZZ/2 
\end{align}
so that \eqref{dwinv} specializes to 
\begin{align}
Z(M, \ZZ/2, \alpha):= \frac{1}{2} \sum_{\rho \in \Hom(\pi_1(M), \ZZ/2)} (-1)^{\langle\hat{\rho}^*(\alpha), [M]\rangle}
\end{align}
If $H^1(M, \ZZ/2) = 0$ then $Z(M, \ZZ/2, \alpha)= \frac{1}{2}$, otherwise $Z(M, \ZZ/2, \alpha)$ is an integer since an even number of terms are involved in the sum. 
\end{example}
One can extend the definiton of these invariants to manifolds with boundary and obtain a 3d TQFT where the invariant $Z(M, G, \omega)$ is the partition function of the theory evaluated on $M$. \\\\

\newcommand{\Diff}{\text{Diff}}
 
Let $S$ be a closed, connected, and oriented surface of genus $g>0$. Let $\Diff(S)^+$ be the group of orientation preserving diffeomorphisms of $S$ and let $\Diff(S)_0^+\subset\Diff(S)^+$ be the connected component of the identity, with respect to the compact open topology. The mapping class group of $S$ is
\begin{align}
    \Gamma(S):= \Diff(S)^+ /\Diff(S)_0^+
\end{align}
It is the group of isotopy classes of orientation preserving diffeomorphims of $S$. For $f$ in $\Gamma(S)$, the mapping torus of $f$ is the three fold 
\begin{align}
   M(f):= S \times [0, 1]/\sim 
\end{align}
where the equivalence $\sim$ is given by
\begin{align}
(x, 0) \sim (f(x), 1)
\end{align}
The three fold $M(f)$ is closed, connected, and orientable
and its diffeomorphism class does not depend on the choice of the representative of the mapping class $f$.  It is the total space of a fiber bundle over the circle 
\begin{align}
\label{fibration}
    \pi: M(f) \to S^1 
\end{align}
such that for all $x$ in $S^1$ the fiber $\pi^{-1}(x)$ is diffeomorphic to the surface $S$ and the image of the monodromy map $\pi_1(S^1) \to \Gamma(S)$ is  generated by $f$. Two mapping tori  $M(f)$ and $M(g)$ are diffeomorphic to each other as fibre bundles if and only if $f$ is conjugate to $g$.
We note that there exist examples of two mapping tori that are diffeomorphic to each other as three manifolds, but not diffeomorphic as fiber bundles.

Let $N(S) \subset M(f)$ be the collar neighbourhood of some fiber of \eqref{fibration}   and let $W=M(f)- \text{Interior}(N(S))$. Since $W$ is diffeomorphic to $S\times I$, the Mayer-Vietoris sequence associated with the decomposition $M(f) = W\cup N(S)$ reduces to
\begin{align}
\label{mv}
    \cdots \to H_1(S, \ZZ) \xrightarrow{\hat{f} - \Id} H_1(S, \ZZ) \to H_1(M(f), \ZZ) \to \ZZ\to 0
\end{align}
where
\begin{align}
    \hat{f} : H_1(S, \ZZ) \to H_1(S, \ZZ) 
\end{align}
is the induced action of $f$ on the homology of $S$, see \cite[Fact 2.1]{Riv}. The sequence \eqref{mv} implies 
\begin{align}
\label{hommf}
    H_1(M(f), \ZZ) \cong \ZZ \oplus Coker(\hat{f} - \Id) 
\end{align}
There exist two matrices $P$ and $Q$ in $\GL(2g, \ZZ)$ such that 
\begin{align}
    P.(\hat{f} - \Id).Q = \begin{bmatrix}
    A(1) &  & &&&&\\
    & A(2) & &&&&\\
    & &\ddots & &&&\\
    & & & A(N)&&&\\
    &&&&0&&\\
    &&&&&\ddots &&\\
    &&&&&&&0
  \end{bmatrix}
\end{align}
where for all $1\leq i \leq N$ the diagonal entries $A(i)$ are  positive integers such that for all $i>1$, $A(i-1)$ divides $A(i)$. The diagonal matrix on the RHS is called the Smith normal form of $(\hat{f} -\Id)$ and will be denoted as $SNF(\hat{f} -\Id)$. The $SNF(\hat{f} -\Id)$ does not depend on the pair $P,Q$ and is an invariant of the conjugacy class of $\hat{f}$ in $\Sp(2g, \ZZ)$. In particular, 
\begin{align}
\label{coker}
    Coker(\hat{f} - \Id) = \ZZ^{2g-N} \oplus \ZZ/A(1) \oplus \cdots \oplus \ZZ/A(N)
\end{align}
The diagonal entries $A(i)$ of $SNF(\hat{f} -\Id)$ can be computed from the the entries of the matrix $(\hat{f} -\Id)$ as follows: Let $D(i)$ be the greatest common divisor of all the determinants of $i \times i$ minors of $(\hat{f} -\Id)$, and let $D(0):=0$,  then 
\begin{align}
\label{algo}
    A(i)= \frac{D(i)}{D(i-1)}          
\end{align}
The fastest known algorithm for computing these integers $A(i)$ has runtime complexity 
\begin{align}
O(\norm{\hat{f}-\Id}\log \norm{\hat{f}-\Id} 16g^4\log 2g) 
\end{align}
\\
Let $G=\ZZ/p$ for $p$ a prime. From \eqref{hommf} and \eqref{coker} it follows that the untwisted  Dijkgraaf-Witten invariant is
\begin{align}
\label{ai}
    Z(M(f), \ZZ/p) 
    &= \frac{|\Hom(H_1(M(f), \ZZ), \ZZ/p)|}{|\ZZ/p|}\\
   &= p^{\big| \{A(i),\, 1\leq i \leq N \,|\, \text{p divides A(i)}\}\big|+(2g-N)}
\end{align}
\\
\begin{example}
\begin{figure}[htp]
\label{fs2}
    \centering
    \includegraphics[width=13cm, height=8cm]{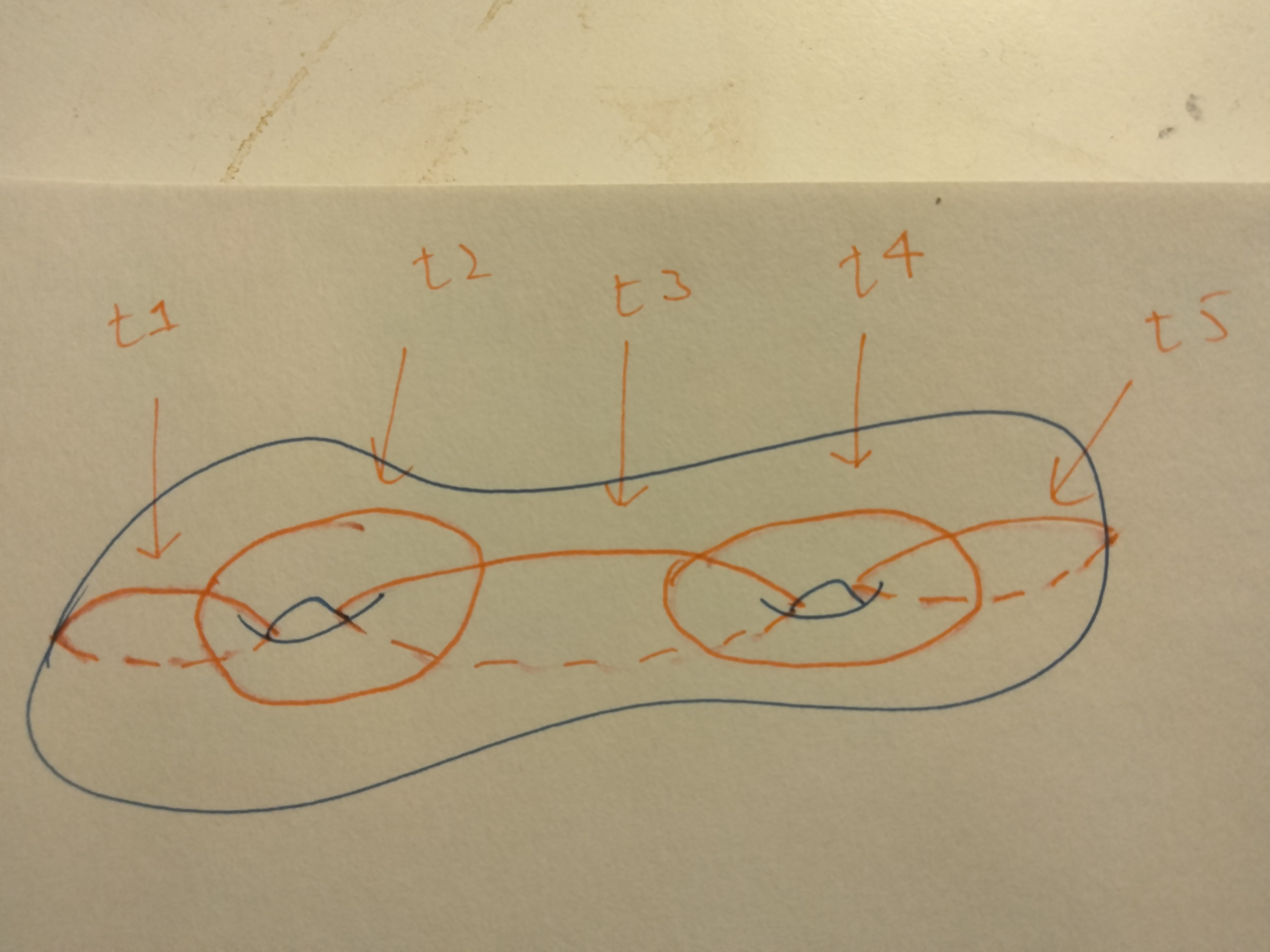}
    \caption{}
\end{figure}
For $1\leq i \leq 5$ let $D^n(t_i)$ denote the $n^{th}$ power of the Dehn twist around the loop $t_i$ on a genus two surface displayed in Fig.1, \eqref{fs2}. If we let 
\begin{align}
f :=D^{39}(t_1)D^{13}(t_4)D^{3}(t_2)D^{19}(t_3)D(t_5)
\end{align}
then
\begin{align}
\hat{f}=\left(
\begin{array}{cccc}
 -116 & -1463 & 39 & -2926 \\
 0 & 1 & 0 & 1 \\
 -3 & -38 & 1 & -76 \\
 0 & -13 & 0 & -12 \\
\end{array}
\right)  \quad\quad\quad SNF(\hat{f}- \Id)= \left(
\begin{array}{cccc}
 1 & 0 & 0 & 0 \\
 0 & 1 & 0 & 0 \\
 0 & 0 & 3 & 0 \\
 0 & 0 & 0 & 507 \\
\end{array}
\right) 
\end{align}
It follows that $H_1(M(f),\ZZ) = \ZZ \oplus \ZZ/3 \oplus \ZZ/507$, and we have
$$Z(M(f), \ZZ/3)= 3^2,\quad Z(M(f), \ZZ/13)= 13,\;\;\text{and}\;\; Z(M(f),\ZZ/7)= 1$$
For another example, let
\begin{align}
f := D(t_1)D^2(t_4) 
\end{align}
then
\begin{align}
\hat{f}=\left(
\begin{array}{cccc}
 1 & 0 & 1 & 0 \\
 0 & 1 & 0 & 0 \\
 0 & 0 & 1 & 0 \\
 0 & -2 & 0 & 1 \\
\end{array}
\right) \quad\quad\quad SNF(\hat{f} - \Id) = \left(
\begin{array}{cccc}
 1 & 0 & 0 & 0 \\
 0 & 2 & 0 & 0 \\
 0 & 0 & 0 & 0 \\
 0 & 0 & 0 & 0 \\
\end{array}
\right)    
\end{align}
It follows that $H_1(M(f), \ZZ) = \ZZ^3 \oplus \ZZ/2$, and
\begin{align}
    Z(M(f), \ZZ/2) = 2^{3}, \;\;\text{and}\;\;  Z(M(f), \ZZ/p) = p^2\,\, \text{for all odd primes $p$}
\end{align}
\end{example}

The duality between 3D TQFT and 2D RCFT allows one to compute $\ZZ/p$ Dijkgraaf-Witten invariants without computing the Smith normal form:
The long exact sequence of homotopy groups associated with the fibration \eqref{fibration} reduces to 
\begin{align}
\label{pi1mf}
    1\to \pi_1(S) \to \pi_1(M(f)) \to \pi_1(S^1) \to 1 
\end{align}
The sequence \eqref{pi1mf} induces an outer action of $\pi_1(S^1)$ on $\pi_1(S)$. This outer action can be refined to an honest action by choosing a section of the sequence, which is possible since $\pi_1(S^1)\cong\ZZ$. With this choice of splitting, $\pi_1(M(f))$ becomes the semidirect product
\begin{align}
    \pi_1(M(f)) = \pi_1(S) \rtimes \langle f\rangle 
\end{align}
The action of $\langle f\rangle$ on $\pi_1(S)$ induces an action on 
$\Hom(\pi_1(S), G)$, and 
\begin{align}
    \Hom(\pi_1(M(f)), G) = \{ (\rho, g) \in \Hom(\pi_1(S), G) \times G \,|\, g.\rho.g^{-1}=f^*(\rho)\}
\end{align}
In particular, when $G$ is an Abelian group the untwisted Dijkgraaf-Witten invariant becomes the number of fixed points of the $f$ action on $\Hom(\pi_1(S), G)$;

\begin{align}
\label{fpabelian}
    Z(M(f), G) &= \frac{|\Hom(\pi_1(M(f)), G)|}{|G|}\\\
    &=|\{\rho \in \Hom(\pi_1(S), G) \,|\, f^*(\rho)=\rho\}|
\end{align}

In case $G=\ZZ/p$ for $p$ a prime, we have that $\Hom(\pi_1(S), \ZZ/p)\cong \Hom((\ZZ/p)^{2g}, \ZZ/p)$ and the action of $f$ on $\Hom(\pi_1(S), \ZZ/p)$ is given by the action of the matrix 
\begin{align}
   ( \hat{f} \,\, \text{mod $p$}) \in \Sp(2g, \ZZ/p) 
\end{align}
on $(\ZZ/p)^{2g}$. Let $\mathcal{F}:= L^2(\Hom((\ZZ/p)^{2g}, \ZZ/p), \C)$ be the space of complex valued $L^2$ functions on $\Hom((\ZZ/p)^{2g}, \ZZ/p)$. $\mathcal{F}$ is a complex vector space of dimension $p^{2g}$,  and the action of $\Sp(2g, \ZZ/p) $ on $\Hom((\ZZ/p)^{2g}, \ZZ/p)$ induces a linear action on this vector space which we denote by 
\begin{align}
\label{rep1}
    \Phi: \Sp(2g, \ZZ/p) \to \Aut(\mathcal{F})=\GL(p^{2g}, \C) 
\end{align}
For every mapping class $f$ in $\Gamma(S)$ this representation has the property:
\begin{align}
\label{trace=fp}
    Trace(\Phi(\hat{f}\,\text{mod} \,\,p)) = |\{\rho \in \Hom((\ZZ/p)^{2g}, \ZZ/p) \,|\, (\hat{f}\,\text{mod}\,p)(\rho)=\rho\}|
\end{align}
Together with \eqref{fpabelian} and \eqref{ai}, this implies the following relationship between the diagonal entries $A(i)$ of the $SNF(\hat{f}-\Id)$ and the representation \eqref{rep1}
\begin{align}
\label{2dsnf}
  \frac{\log Trace(\Phi(\hat{f}\,\text{mod} \,\,p))}{\log p} ={\big| \{A(i),\, 1\leq i \leq N \,|\, \text{p divides A(i)}\}\big|+(2g-N)}
\end{align}

\begin{remark}
With a rather clever use of two cocycles, it is shown in \cite{DW} that for any finite group $G$ and a class $\omega$ in $H^3(BG, U(1))$, there exists a space of complex valued wave functions on  $\Hom(\pi_1(S), G)/G$ denoted as  $\mathcal{F}^{(\omega)}$, such that the induced representation  of the mapping class group 
\begin{align}
    \Phi^{(\omega)} : \Gamma(S) \to \Aut(\mathcal{F}^{(\omega)}  )
\end{align}
satisfies 
\begin{align}
    Trace(\Phi^{(\omega)}(f)) = Z(M(f), G, \omega)
\end{align}
\end{remark}

\section{$\ZZ/p$ Dijkgraaf-Witten invariants of 3D mapping tori of genus one}
Let $S$ be a closed, connected, and oriented surface of genus $g=1$. The mapping class group $\Gamma(S)$ is isomorphic to $\SL(2, \ZZ)$. If $A = \begin{pmatrix} a & b\\ c & d \end{pmatrix}$ in $\SL(2, \ZZ)$, then the entries of the Smith normal form 
\begin{align}
\label{snf}
    SNF(A-\Id) = \begin{pmatrix} A(1) & 0 \\ 0 & A(2) \end{pmatrix}
 \end{align}
 can be calculated using \eqref{ai},  and are given as  
\begin{align}
\label{a1a2}
    A(1) = \GCD (a-1, b, c, d-1) \quad \quad \text{and} \quad \quad A(2) =\frac{ |\Tr(A) - 2|}{A(1)}  
\end{align}
Since $2- \Tr(A) = \Det(A - \Id)$, it follows that $|\Tr(A) -2|$ is divisible by $A(1)^2$, confirming that $A(2)$ is an integer which is divisible by $A(1)$. Denoting by $M(A)$ the mapping torus of $A$,  formula \eqref{a1a2} along with \eqref{hommf} and \eqref{coker} implies the following trichotomy;  

\begin{align}
    H_1(M(A), \ZZ) = \begin{cases} \ZZ \oplus \ZZ/A(1) \oplus \ZZ/A(2) & \text{if $\Tr(A) \neq$ }2 \\
    \ZZ^2 \oplus \ZZ/A(1)  & \text{if $\Tr(A) =2$ and $A\neq \Id$ }\\
    \ZZ^3 & \text{if $A= \Id$}\end{cases} 
\end{align}
The following proposition gives a formula for $Z(M(A), \ZZ/p)$, the untwisted Dijkgraaf-Witten invariant, defined in \eqref{dwuntwisted}.
\begin{proposition}
\label{dwformula}
    Let $p$ be a prime and let $A$ be in $\SL(2, \ZZ)$,  then 
\begin{align}
\label{prop}
Z(M(A), \ZZ/p) =\begin{cases}p^2  & \text{if $A \equiv \Id \mod p $}\\
p  & \text{if $A \nequiv \Id \mod p $ and $\Tr(A) \equiv 2 \mod p$}\\
1 & \text{otherwise}
\end{cases}
\end{align}
\end{proposition}

\begin{proof}
    Assume that $\Tr(A) \neq 2$. Then both $A(1)$ and $A(2)$ are non-zero and   
    \begin{align}
\hash \{\Hom (\ZZ/A(1)\oplus \ZZ/A(2), \ZZ/p)\} =\begin{cases} p^2 & \text{if $p$ divides both $A(1)$ and $A(2)$} \\ p  & \text{if $p$ divides $A(1)$ or $A(2)$ } \\ 1 & \text{if $p$ does not divide $A(1)$ and $A(2)$}
\end{cases}
\end{align}

 \emph{If $A \equiv \Id \mod p$ then $p$ divides both $A(1)$ and $A(2)$}. This follows since if $A \equiv \Id \mod p$ then all four integers $a-1, b, c, d-1$ are divisible by $p$. Since A(1) is the GCD of these four integers  it follows that $A(1)$ is divisible by $p$. Since $A(2)$ is always divisible by $A(1)$, it follows that $A(2)$ is also divisible by $p$.

 \emph{If $A \nequiv \Id \mod p$ and $\Tr(A) \equiv 2 \mod p$ then $p$ divides $A(2)$ but does not divide $A(1)$}. This follows since if $A \nequiv \Id \mod p$ then $p$ does not divide $A(1)$. Since $p$ does not divide $A(1)$ and  $\Tr(A) \equiv 2 \mod p$ then from definition of A(2) in \eqref{a1a2} it is clear that $p$ divides A(2).

In all other cases, $p$  divides neither $A(1)$ nor $A(2)$. This proves \eqref{prop} in the case $\Tr(A)\neq 2$.  

Assume $\Tr(A) =2$ and $A\neq \Id$. In this case the conjugacy class of $A$ contains a unique upper triangular matrix $\begin{pmatrix} 1 & n \\ 0 & 1\end{pmatrix}$. Since $A(1)$ and $A(2)$ are invariants of the conjugacy class of $A$,  it follows that $A(1)=n$ and $A(2)=0$. This implies  
\begin{align}
\label{prop1}
\hash \{\Hom (\ZZ\oplus \ZZ/A(1), \ZZ/p)\}  =\begin{cases}p^2  & \text{if $n \equiv 0 \mod p $}\\
p  & \text{otherwise}
\end{cases}
\end{align}
But if $n\equiv 0 \,\,mod\,p$ then $A \equiv \Id \,mod\,p$. If $A=\Id$, then $M(A)$ is the 3-torus and it is clear that $Z(M(A), \ZZ/p)= p^2$. 
\end{proof}

In light of the 3D/2D correspondence given by \eqref{fpabelian}, the corresponding 2D statement and proof are as follows:   
 \begin{proposition}
\label{fixedpointformula}
Let $p$ be a prime and let $A$ be in  $\SL(2, \ZZ)$, then 
\begin{align}
    \hash \{ \rho \in \Hom(\ZZ^2, \ZZ/p) \,|\, A(\rho) = \rho\} = \begin{cases}p^2  & \text{if $A \equiv \Id \mod p $}\\
p  & \text{if $A \nequiv \Id \mod p $ and $\Tr(A) \equiv 2 \mod p$}\\
1 & \text{otherwise}
\end{cases}
\end{align}
\end{proposition}

\begin{proof} The $\SL_2(\ZZ)$-action on $\Hom(\ZZ^2,\ZZ/p)$ factors through the reduction mod $p$ map $\SL_2(\ZZ)\rightarrow\SL_2(\FF_p)$. Since $\ZZ/p$ is $p$-torsion, any map $\ZZ^2\rightarrow\ZZ/p$ factors uniquely through $\FF_p^2$, so this gives a bijection $\Hom(\ZZ^2,\ZZ/p)\stackrel{\sim}{\longrightarrow}\Hom_{\FF_p}(\FF_p^2,\FF_p)$, and the induced action of $\SL_2(\FF_p)$ on $X := \Hom_{\FF_p}(\FF_p^2,\FF_p)$ agrees with the action induced by the standard action on $\FF_p^2$. There is also a free action of $\FF_p^\times$ on $X$ (by acting on the codomain), which commutes with the $\SL_2(\FF_p)$-action.

Let $\overline{A}\in\SL_2(\FF_p)$ be image of $A$. Since any map to $\FF_p$ is either zero or surjective, it suffices to count the points in $X$ fixed by $\overline{A}$. Given a surjection $\varphi : \FF_p^2\rightarrow\FF_p$, $\overline{A}$ preserves the $\FF_p^\times$-orbit of $\varphi$ if and only if it preserves $\ker\varphi$, which is a line in $\FF_p^2$. Since the action of $\overline{A}$ commutes with that of $\FF_p^\times$, $\overline{A}$ fixes $\varphi$ if and only if it fixes every element in the $\FF_p^\times$ orbit of $\varphi$, which is to say that it induces, via $\varphi$, the identity on $\FF_p$. Thus for any $v\in\varphi^{-1}(1)$, $\varphi\circ \overline{A} = \varphi$ if and only if $\overline{A}$ preserves the kernel and $\overline{A}v \equiv v\mod \ker\varphi$, which is to say that $\overline{A}$ is conjugate to a matrix of the form $\begin{psmallmatrix} a & b \\ 0 & 1\end{psmallmatrix}$, and since $\det\overline{A} = 1$, we must have $a = 1$, or equivalently, that $\tr\overline{A} = 2$. Thus, if such an $\overline{A}$ is not the identity and has a fixed point, then it preserves at most one line in $\FF_p^2$, and fixes every point in the corresponding $\FF_p^\times$-orbit, which is to say it has exactly $p-1$ fixed points. Since $|X| = p^2-1$ and the zero map $\FF_p^2\rightarrow \FF_p$ is fixed by everything, this implies the desired result.
    
\end{proof}

From \eqref{trace=fp} it follows that $Z(M(A), \ZZ/p)$ only depends on the conjugacy class of $A\mod p$ in $\SL(2, \ZZ/p)$. The conjugacy classes of the group $\SL(2, \ZZ/p)$ are described in Table \ref{tab}, due essentially to Frobenius. Here  $p$ is an odd prime, $x$ is in $(\ZZ/p)^*$, $x^{-1}$ is the multiplicative inverse of $x$, and $\epsilon$ is a generator of $(\ZZ/p)^*$.

\begin{table}
\begin{tabular}{ |c|c|c|c| } 
\label{conjclasses}
  Representative  & $\hash$ elements in class & $\hash$ classes & Label \\ 
 $\begin{pmatrix} 1 & 0\\0&1\end{pmatrix}$ & 1 & 1 & $\mathcal{C}1$ \\ 
 $\begin{pmatrix} -1 & 0\\0&-1\end{pmatrix}$ & 1 & 1 & $\mathcal{C}2$ \\ 
$\begin{pmatrix} 1 & 1\\0&1\end{pmatrix}$ & $\frac{p^2-1}{2}$ & 1 & $\mathcal{C}3$ \\
$\begin{pmatrix} 1 & \epsilon\\0&1\end{pmatrix}$ & $\frac{p^2-1}{2}$ & 1 & $\mathcal{C}4$ \\
$\begin{pmatrix} -1 & 1\\0&-1\end{pmatrix}$ & $\frac{p^2-1}{2}$ & 1 & $\mathcal{C}5$ \\
$\begin{pmatrix} -1 & \epsilon\\0&-1\end{pmatrix}$ & $\frac{p^2-1}{2}$ & 1 & $\mathcal{C}6$ \\
$\begin{pmatrix} x & 0\\0&x^{-1}\end{pmatrix}, \,x\neq\pm1$ & $p(p+1)$ & $\frac{p-3}{2}$ & $\mathcal{C}7$ \\
$\begin{pmatrix} x & y\\\epsilon y&x\end{pmatrix},\,x\neq \pm 1$ & $p(p-1)$ & $\frac{p-1}{2}$ & $\mathcal{C}8$ 
\end{tabular}
\caption{The conjugacy classes of $\SL(2,\Z/p)$}
\label{tab}
\end{table}

Using this table we can prove the following, 
\begin{corollary}
\label{cor}
    Let $p$ be an odd prime and let $[A]$ be a conjugacy class in $\SL(2, \ZZ)$,  then 
\begin{align}
\label{prop}
Z(M([A]), \ZZ/p) =\begin{cases}p^2  & \text{if $[A]\bmod p$ is the identity class $\mathcal{C}1$}\\
p  & \text{if $[A] \bmod p $ a unipotent class, that is, $\mathcal{C}3$ or $\mathcal{C}4$}\\
1 & \text{otherwise}
\end{cases}
\end{align}
\end{corollary}
\begin{proof}
   Notice that $x+x^{-1} =2$ is only solved by $x=1$; indeed, multiply by $x$ to obtain
    \begin{align}
        x^2 - 2x +1 = 0,
    \end{align}
    and observe that the left hand side factors as $(x-1)^2$. It is now easily checked that only the representatives of $\mathcal{C}3$ and $\mathcal{C}4$ whose trace mod $p$ is 2. The corollary is now implied by Propositions \ref{dwformula} and \ref{fixedpointformula}. 
\end{proof}

\begin{example} Let $p=3$. We take  $A \in \SL(2, \ZZ)$ such that $A \equiv \Id \mod 3$ and notice that both $A(1)$ and $A(2)$ are divisible by $3$
\begin{align}
    A=\begin{pmatrix} 10 & 3 \\ 
 3 & 1 \end{pmatrix}     \quad\quad\quad SNF(A -ID)= \begin{pmatrix} 3 & 0 \\ 
 0 & 3 \end{pmatrix}
\end{align}
The powers of $A$ are also congruent to $\Id$ $\mod 3$, we see that the Smith normal form invariants of powers are also divisible by $3$  
\begin{align}
    A^2=\begin{pmatrix} 109 & 33 \\ 
 33 & 10 \end{pmatrix}     \quad\quad\quad SNF(A^2 -ID)= \begin{pmatrix} 3 & 0 \\ 
 0 & 39 \end{pmatrix}
\end{align}

\begin{align}
    A^3=\begin{pmatrix} 1189 & 360 \\ 
 360 & 109 \end{pmatrix}     \quad\quad\quad SNF(A^3 -ID)= \begin{pmatrix} 36 & 0 \\ 
 0 & 36 \end{pmatrix}
\end{align}
Following is an example of $A \in \SL(2, \ZZ)$ such that $A\nequiv \Id \mod 3$ and $\Tr(A)\equiv 2 \mod 3$. We see that $A(1)$ is not divisible by $3$ but $A(2)$ is divisible by $3$. 
\begin{align}
    A=\begin{pmatrix} 4 & 3 \\ 
 5 & 4 \end{pmatrix}     \quad\quad\quad SNF(A -ID)= \begin{pmatrix} 1 & 0 \\ 
 0 & 6 \end{pmatrix}
\end{align}
For the last case we take $A \in \SL(2, \ZZ)$ which does not satisfy the mod $3$ conditions and see that $3$ does not divide $A(1)$ and $A(2)$. 
\begin{align}
    A=\begin{pmatrix} 27 & 1 \\ 
 -1 & 0 \end{pmatrix}     \quad\quad\quad SNF(A -ID)= \begin{pmatrix} 1 & 0 \\ 
 0 & 25 \end{pmatrix}
\end{align}
\end{example}

\section{Counting estimates and Proofs of Theorems \ref{thm:mainasymptotics} and \ref{thm:snfasymptotics}}

Let $G=\SL(2,\Z/p\Z)$ (or any finite group) and let $\rho : SL(2,\Z)\twoheadrightarrow G$ be an epimorphism. 
A conjugacy class of $\SL(2,\ZZ)$ is hyperbolic if its trace has absolute value $> 2$, and is primitive if it is not a power any other conjugacy class.
To each primitive hyperbolic conjugacy class $[A]$ of $\SL(2,\ZZ)$, let $C([A])$ denote the corresponding conjugacy class in $G$ which is the image under $\rho$ of $[A]$. Then Sarnak's Chebotarev Density Theorem for Prime Geodesics \cite[Theorem 3.16]{Sar}  states that, for any fixed conjugacy class $C$ in $G$, we have:
\begin{align}
\label{asymptotics}
\sum_{[A]\in[\SL(2,\Z)],|\Tr([A])| < T\atop C([A]) = C}1
=
{|C|\over |G|}
\li(T^2)
+
O(T^{3/2}(\log T)^2),
\end{align}
as $T\to\infty$.
Here $\li(x)$ is the logarithmic integral function,
$$
\li(x)=\int_2^x{dt\over \log t}.
$$
Note that the above is a general phenomenon not restricted to $\Gamma=\SL(2,\Z)$ and $G=\SL(2,\Z/p\Z)$; in the general case, there may be further ``main'' terms, corresponding to Laplace-Beltrami eigenvalues $\lambda_j=s_j(1-s_j)$ with $1/2<s_j<1$, see Sarnak op. cit. In our case of the modular group $\SL(2,\Z)$ and its principal congruence (normal) subgroup $\Gamma(p)=\ker(\SL(2,\Z)\twoheadrightarrow G)$, 
we have by Selberg's 3/16ths theorem \cite{Sel}\footnote{Today even better bounds are known towards the generalized Ramanujan conjectures \cite{KS}, which in this setting at infinity amount to the Selberg's 1/4 conjecture. But Selberg's result is already sufficient for our purposes here.} that $\lambda_j\ge3/16,$ that is, $s_j\le 3/4$. Therefore any lower order terms, which are of the order $T^{2s_j}$, are already dominated by the error term $T^{3/2}$.

We can now prove the main theorems.
\begin{proof}[Proof of Theorems \ref{thm:mainasymptotics} and \ref{thm:snfasymptotics}]
    Theorem \ref{thm:mainasymptotics} follows from \eqref{asymptotics}, Corollary \ref{cor}, the values in Table \ref{conjclasses} and the fact that 
    \begin{align}
        |\SL(2, \ZZ/p)| = p^3-p .
    \end{align}

    Theorem \ref{thm:snfasymptotics} now similarly follows  from Sarnak's Chebotarev density theorem \eqref{asymptotics}, together with the previous results on Smith normal forms of matrices $A- \Id$ for $A$ in $\SL(2, \ZZ)$. 
\end{proof}

\section{$\ZZ/2$ Dijkgraaf-Witten invariants of 3D mapping tori of genus one and modular forms}\label{sec:2.3}
The Table \ref{tab}, Corollary \ref{cor}, and Theorem \ref{thm:mainasymptotics} are all valid for $p$ an odd prime; we now extend these to $p=2$: Since $ (\ZZ/2)^2$ has three distinct subgroups of order two it follows that 
\begin{align}
    \hash \Hom((\ZZ/2)^2, \ZZ/2)=4 
\end{align}
Label these four elements as  
\begin{align}
    \{\rho_0, \rho_1, \rho_2, \rho_3\} = \Hom((\ZZ/2)^2, \ZZ/2)
\end{align}
where $\rho_0$ is the trivial homomorphism, i.e. 
\begin{align}
    Ker (\rho_0: (\ZZ/2)^2 \to \ZZ/2 ) = (\ZZ/2)^2
\end{align}
and the other three elements are labelled arbitrarily. The action of the group $\SL(2, \ZZ/2)$ on $\Hom((\ZZ/2)^2, \ZZ/2)$ fixes the trivial homomorphism $\rho_0$ and permutes the other three homomorphisms $\{\rho_1, \rho_2, \rho_3\}$ amongst themselves. The group $\SL(2, \ZZ/2)$ has order six and its six elements are 
\begin{align}
  \SL(2, \ZZ/2):= \Bigg\{ &\Id= \begin{pmatrix}
      1 & 0\\ 0 &1 
  \end{pmatrix}, \,   T=\begin{pmatrix} 1 & 1\\ 0 &1 \end{pmatrix},\,
      S= \begin{pmatrix} 0 & 1\\ -1 &0 \end{pmatrix},\,\\
      &T.S=\begin{pmatrix} 1 & -1\\ 1 &0 \end{pmatrix},\,
      T.S.T=\begin{pmatrix} 1 & 0\\ 1 &1 \end{pmatrix},\,
      -S.T^{-1}=\begin{pmatrix} 0 & 1\\ -1 &1 \end{pmatrix} \Bigg\}
\end{align}
It has three conjugacy classes which are composed of 
\begin{align}
    &\mathcal{C}1:= \{ \Id \} \\
    &\mathcal{C}2:=\{T, \, S, \, T.S.T \} \\
    &\mathcal{C}3:=\{T.S,\, -S.T^{-1}\}
\end{align}
Elements in $\mathcal{C}2$ have order 2 and elements in $\mathcal{C}3$ have order 3.
The group $\SL(2, \ZZ/2)$ is isomorphic to $\mathfrak{S}_3$, the symmetric group on three letters, and every possible permutation of $\{\rho_1, \rho_2, \rho_3\}$ occurs under the action of $\SL(2, \ZZ/2)$. The following table shows the permutations assigned to  elements of $\SL(2, \ZZ/2)$ under this isomorphism: 
\begin{align}
\label{symaction}
\begin{tabular}{c c c}
\hline
$\SL(2, \ZZ/2)$ & $\mathfrak{S}_3$ & permutation\\
    $T$ &   (12)  & $\{\rho_1, \rho_2, \rho_3\}\mapsto \{\rho_2, \rho_1, \rho_3\}$\\ 
    $S$ & (23) & $\{\rho_1, \rho_2, \rho_3\}\mapsto \{\rho_1, \rho_3, \rho_2\}$\\
    $T.S$ & (123) & $\{\rho_1, \rho_2, \rho_3\}\mapsto \{\rho_3, \rho_1, \rho_2\}$\\
    $T.S.T$ & (13) &  $\{\rho_1, \rho_2, \rho_3\} \mapsto \{\rho_3, \rho_2, \rho_1\}$\\
    $-S.T^{-1}$ & (132) & $\{\rho_1, \rho_2, \rho_3\} \mapsto \{\rho_2, \rho_3, \rho_1\}$
\end{tabular}
\end{align}
We see that a non-trivial element $A$ of $\SL(2, \ZZ/2)$ fixes at most a single element of the set $\{\rho_1, \rho_2, \rho_3\}$, which happens if and only if $A$ is in the conjugacy class $\mathcal{C}2$. 
\begin{corollary}
    Let $[A]$ be a conjugacy class in $\SL(2, \ZZ)$,  then 
\begin{align}
Z(M([A]), \ZZ/2) =\begin{cases}4  & \text{if $[A]\bmod 2$ is the identity class $\mathcal{C}1$}\\
2  & \text{if $[A] \bmod 2 $ is the class $\mathcal{C}2$}\\
1 & \text{otherwise}
\end{cases}
\end{align}
\end{corollary}
\begin{proof}
    From \eqref{fpabelian} we have that $Z(M(A), \ZZ/2)$ is equal to the number of fixed points of $A\,\bmod\,2$ acting on $\Hom((\ZZ/p)^2, \ZZ/p)$. From the previous discussion we know that if $A \bmod 2$ is not the identity, then it has two fixed points if and only if it belongs to the class $\mathcal{C}2$, and one fixed point if and only if it belongs to $\mathcal{C}3$. The result follows. 
\end{proof}

Let
\begin{align}
    \Gamma(2):= Ker( \SL(2, \ZZ) \to \SL(2, \ZZ/2)
\end{align}
the principal congruence subgroup of level two, and 
\begin{align}
 X(2) = \HH/\Gamma(2)
\end{align}
the modular curve of level two.
The curve $X(2)$ has genus zero, zero elliptic points, and three cusps. Let 
\begin{align}
    \lambda: \HH \to \C 
\end{align}
denote the modular lambda function where
\begin{align}
\label{lambda}
    \lambda(\tau) =  \frac{e_3(\tau) - e_2(\tau)}{e_1(\tau) - e_2(\tau)}
\end{align}
and 
\begin{align}
\label{ei}
 e_1(\tau) = \wp(\tau/2, \tau), \quad e_2(\tau)=\wp (1/2, \tau),   \quad \quad e_3(\tau)= \wp((1+\tau)/2, \tau) 
\end{align}
with 
\begin{align}
    \wp(z, \tau) = \frac{1}{z^2} \,+ \sum_{(m,n) \in \ZZ^2 -\{0, 0\}} \bigg( \frac{1}{(z-m\tau-n)^2} - \frac{1}{(m\tau +n)^2}  \bigg) 
\end{align}
The function \eqref{lambda} is invariant under the $\Gamma(2)$ action on $\HH$ and descends to a holomorphic function on the quotient 
\begin{align}
    \lambda: X(2) \to \C-\{0, 1\} 
\end{align}
\begin{lemma}
\label{cool}
    Let $\zeta_3:=e^{2 \pi i /3}$ and let $[A]$ be a conjugacy class in   $[\SL(2, \ZZ)]$ such that $[A] \bmod 2$ is not the identity class. Then 
    \begin{align}
     Z(M(A), \ZZ/2) =  \bigg|  \frac{3}{2\pi i} \log \lambda ( A (\zeta_3)) \bigg|^{-1} 
    \end{align}
and if $[A] \bmod 2$ is the identity class, then 
    \begin{align}
       Z(M(A), \ZZ/2) =4
    \end{align}
\end{lemma}

The following values, which are calculated using a computer algebra program, give numerical evidence for Lemma \ref{cool}:
\begin{align}
\lambda(\zeta_3)&=0.5\,-0.866025i \,\,\approx -e^{\frac{2\pi i}{3}} \\
\lambda(T(\zeta_3))&=0.5+0.866025i \,\,\approx e^{\frac{2\pi i}{6}}\\
\lambda(S(\zeta_3))&=0.5+0.866025i\,\,\approx e^{\frac{2\pi i}{6}}\\
\lambda(T.S(\zeta_3))&=0.5-0.866025i\,\,\approx -e^{\frac{2\pi i}{3}}\\
\lambda(T.S.T(\zeta_3))&=0.5+0.866025i\,\,\approx e^{\frac{2\pi i}{6}}\\
\lambda(-S.T^{-1}(\zeta_3))&=0.5-0.866025i\,\,\approx -e^{\frac{2\pi i}{3}}
\end{align}

\begin{proof}
The functions $\{e_1, e_2, e_3\}$ defined in \eqref{ei} while invariant under the action of $\Gamma(2)$ are transformed into each other under the action of $\SL(2, \ZZ/2)$, i.e. for any $A$ in $\SL(2, \ZZ/2)$ we have 
\begin{align}
    e_i(A(\tau))=e_j(\tau) .
\end{align}
 If we identify $\{e_1, e_2, e_3\}$ with $\{\rho_1, \rho_2, \rho_3\}$ as a set, then the action of $\SL(2, \ZZ/2)$ on  $\{e_1, e_2, e_3\}$ is precisely the action of $\SL(2, \ZZ/2)$ on $\{\rho_1, \rho_2, \rho_3\}$ as outlined in \eqref{symaction}. 

Note that the function \eqref{lambda} is a cross-ratio of the three functions $\{ e_1, e_2, e_3\}$ and under the action by any $A$ in $\SL(2, \ZZ/2)$ we obtain a new cross-ratio of the same three functions. It is well known that there are in total six different cross-ratios and they are related to each other under  the action of the anharmonic group which is isomorphic to $\SL(2, \ZZ/2)$. In other words, the group $\SL(2, \ZZ/2)$ also acts on the codomain of the  map
\begin{align}
\label{lambda1}
    \lambda: X(2) \to \C-\{0, 1\}
\end{align}
by the following assignment 

\begin{align}
&T\quad\quad\quad\quad \lambda(\tau) \mapsto \frac{\lambda(\tau)}{\lambda(\tau)-1}\\
&S\quad\quad\quad\quad\lambda(\tau) \mapsto 1-\lambda(\tau)
\\
&T.S\quad\quad\quad\quad\lambda(\tau) \mapsto \frac{\lambda(\tau)-1}{\lambda(\tau)}\\
&T.S.T\quad\quad\quad\quad\lambda(\tau) \mapsto \frac{1}{\lambda(\tau)}\\
-&S.T^{-1}\quad\quad\quad\quad\lambda(\tau) \mapsto \frac{1}{1-\lambda(\tau)}\\
\end{align}
It is easy to see that the point $-e^{\frac{2\pi i}{3}}$ in the codomain of the map \eqref{lambda1} is fixed under the action of the elements of order three and is mapped to $e^{\frac{2\pi i}{6}}$ under the action of elements of order two. This proves Lemma \ref{cool}. 
\end{proof}

\begin{lemma}
\label{cool1}
Let $d>1$ be a cube free integer, let $N=3\prod_{p|d}p$, and let $\psi$ denote the quadratic character of conductor three. There exists a modular form of weight one and level $3N^2$ twisted by $\psi$, 
\begin{align}
    f_{\psi} (\tau) \in M_1(3N^2, \psi) 
\end{align}
such that its $p^{th}$ Fourier coeffecient $a_p(f_{\psi})$ satisfies \\
\begin{align}
    a_p(f_{\psi}) = Z(M(A), \ZZ/2) +2   
    \end{align}
    where 
    \begin{align}
  &\text{if $p\equiv 1\,\, mod\, 3$ and $d$ is a cube $mod\, p$ then $[A]\bmod 2$ is the identity class} \\
   &\text{if $p \equiv 1\,\, mod\, 3$ and $d$ is not a cube $mod\, p$ then $[A] \bmod 2$ is in the class $\mathcal{C}2$}\\ 
   &\text{if $p \equiv 2 \,\,mod\, 3$ then $[A] \bmod 2$ is in the class $\mathcal{C}3$}
\end{align}
\end{lemma}

For $d=2$, the first hundred Fourier coeffecients, taken from LMFDB, are as follows:
\begin{align}
f_{\psi}(\tau)= q - q^7 - q^{13} - q^{19} + q^{25} + 2  q^{31} - q^{37} + 2  q^{43} - q^{61} - q^{67} - q^{73} - q^{79} + q^{91} - q^{97}+....
\end{align}

\begin{proof}
    For $d>1$ a cube free integer, the Galois group of the field extension 
\begin{align}
\QQ(d^{1/3}, \zeta_3)
\end{align}
over $\QQ$ is 
\begin{align}
    \Gal(\QQ(d^{1/3}, \zeta_3)|\QQ) \cong \mathfrak{S}_3
\end{align}
and there is a representation 
\begin{align}
    \Psi : \Gal(\QQ(d^{1/3}, \zeta_3)|\QQ) \to \GL(2, \ZZ)
\end{align}
given by the assignment 
\begin{align}
    (123) \mapsto &\begin{pmatrix} 0 & 1\\ -1 & -1 \end{pmatrix} \\
    (23)   \mapsto &\begin{pmatrix} 0 & 1\\ 1 & 0 \end{pmatrix}
\end{align}
such that for any rational prime $p \in \QQ$ and a maximal ideal $\mathfrak{p}$ of the ring of integers of $\QQ(d^{1/3}, \zeta_3)$ lying over $p$, we have that 
 
\begin{align}
    Trace (\Psi(Frob_{\mathfrak{p}})) = \begin{cases} 2& \text{if $p\equiv 1 \,\,mod\,3$ and d is a cube  mod $p$},\\ -1& \text{if $p\equiv 1\,\, mod\, 3$ and d is not a cube mod $p$} \\ 0& \text{if $p \equiv 2\,\,mod\,3$} \end{cases}
\end{align}
It is well known that the representation $\Psi$ is modular with the modular form as assigned in the statement, see for example,  \cite[p. 371]{DS}. This proves \eqref{cool1}.
\end{proof}
The geometric reasoning behind lemma \eqref{cool1} is as follows:
Let $\mathcal{O}_{\QQ(d^{1/3}, \zeta_3)}$ denote the ring of integers of $\QQ(d^{1/3}, \zeta_3)$, then for any rational prime $p$ one has that 
\begin{align}
    p\mathcal{O}_{\QQ(d^{1/3}, \zeta_3)} = \begin{cases} 
\mathfrak{p}_1\cdots \mathfrak{p}_6\quad\quad\quad\quad &\text{if $p\equiv 1\, mod\,\,3$ and $d$ is a cube mod $p$}\\
\mathfrak{p}_1 \mathfrak{p}_2\quad\quad\quad\quad\quad &\text{if $p\equiv 1\, mod\,\,3$ and $d$ is not a cube mod $p$}\\
\mathfrak{p}_1\mathfrak{p}_2\mathfrak{p}_3\quad\quad\quad \quad&\text{if $p\equiv 2\, mod\,\,3$}
    \end{cases}
\end{align}

A closed geodesic $\gamma$ in  $X(1)$ can be pulled back by the covering map 
\begin{align}
  \pi:   X(2)\to X(1)
\end{align}

The pullback $\pi^*(\gamma)$ is a disjoint union of closed geodesics in $X(2)$ which we shall denote as 
\begin{align}
    \pi^*(\gamma)=\{\tilde{\gamma}_1, \dots, \tilde{\gamma}_k\}  
\end{align}
With every closed geodesic $\gamma$ in $X(1)$ is associated a unique conjugacy class $[A_{\gamma}]$  in $[\SL(2, \ZZ)]$. From \cite[Prop 3.11]{Sar}, we have the following 
\begin{align}
    \pi^*(\gamma)=\begin{cases}
      \{\tilde{\gamma}_1, \dots, \tilde{\gamma}_6\} \quad\quad\quad &\text{if $[A_{\gamma}] \bmod 2$ is the identity class $\mathcal{C}1$}   \\
      \{\tilde{\gamma}_1, \tilde{\gamma}_2\} \quad\quad\quad &\text{if  $[A_{\gamma}] \bmod 2 $ is in conjugacy class $\mathcal{C}3$}  \\
        \{\tilde{\gamma}_1, \tilde{\gamma}_2, \tilde{\gamma}_3\} \quad\quad\quad &\text{if $[A_{\gamma}]\bmod 2$ is in conjugacy class $\mathcal{C}2$} 
    \end{cases}
\end{align}
\\\\

We know not how to generalize lemmas \eqref{cool} and \eqref{cool1} to odd primes, but we remark: The function field of the Riemann surface $X(2)$ is 
 \begin{align}
     \C(X(2))=\C(j, e_1, e_2, e_3)   
 \end{align}
where $j :\HH \to \C$ is the modular $j$ function which can be given as 
\begin{align}
    j(\tau) = \frac{1}{54} \frac{\big( (e_1 - e_2)^2 + (e_2 -e_3)^2 + (e_3 - e_1)^2 \big)^3}{\big((e_1 - e_2)(e_2 -e_3)(e_3 - e_1)\big)^2}
\end{align}
We have a Galois covering of Riemann surfaces 
\begin{align}
\label{pi}
   \pi:  X(2) \to X(1):=\HH/\SL(2, \ZZ)
\end{align}
where the function field $\C(X(1)) = \C(j)$. It follows that $\C(j, e_1, e_2, e_3) | \C(j) $ is a Galois field extension with the Galois group 
\begin{align}
    \Gal \big(  \C(j, e_1, e_2, e_3) | \C(j)   \big) \cong \SL(2, \ZZ/2) 
\end{align}
Moreover,  $\{\rho_0, \rho_1, \rho_2, \rho_3\} = \Hom((\ZZ/2)^2 , \ZZ/2)$ and $\{j, e_1, e_2, e_3\}$ are naturally isomorphic as sets with $\SL(2, \ZZ/2)$ action. This isomorphism generalizes to odd $p$, i.e. the action of the Galois group $\SL(2, \ZZ/p)$ on the extension of the function field $\C(j)$ to the function field of $X(p):=\HH/\Gamma(p)$ associated with the Galois covering 
\begin{align}
    \pi(p) : X(p) \to X(1)
\end{align}
can be identified with the action of $\SL(2, \ZZ/p)$ on $\Hom((\ZZ/p)^2, \ZZ/p)$.

\section{Further remarks: $\SU(2)$ Chern-Simons-Witten invariants of mapping tori of genus one}

Let $A$ be in  $\SL(2, \ZZ)$ such that $|\Tr(A)|>2$. One maps $A$ to an indefinite binary quadratic form of discriminant $\Tr(A)^2-4$  as follows 
\begin{align}
\label{ccc} 
    A=\begin{pmatrix} a & b\\ c & d \end{pmatrix} \mapsto bx^2 + (a - d) xy -cy^2 =: Q_A(x, y)
\end{align}
This map is equivariant with respect to the $\SL(2, \ZZ)$ action on the left by conjugation and on the right by the standard action on binary quardatic forms, and one gets an isomorphism between conjugacy classes of hyperbolic elements of $\SL(2, \ZZ)$  with trace $t$ and isomorphism classes of indefinite binary quadratic forms of discriminant $t^2-4$. Under this isomorphism,  if $mx^2+lxy+ky^2$ is a BQF of discriminant $t^2-4$, then its inverse is 
\begin{equation}
    mx^2+lxy+ky^2 \mapsto \begin{pmatrix} \frac{t-l}{2} & k\\ -m & \frac{t+l}{2} \end{pmatrix}
\end{equation}
which has trace equal to $t$. 
Let $[A]$ be in  $[\SL(2, \ZZ)]$ such that $|\Tr([A])|>2$ and 
 let $k$ be a positive integer. Define 
\begin{align}
\label{csw}
    Z(M([A]), \SU(2), k) := \frac{\text{sign}(\Tr(A))}{2}\Bigg(&\frac{1}{\sqrt{|\Tr(A) -2 |}} \sum_{\{(x,y) \in \ZZ^2/(A-\Id)\ZZ^2\}}  
         e^{2\pi i k  \frac{Q_A(x, y)}{\Tr(A) -2}}\,\, +\,\, \\
         & \frac{1}{\sqrt{|\Tr(A) +2 |}} \sum_{\{(x,y) \in \ZZ^2/(-A-\Id)\ZZ^2\}}  
         e^{2\pi i k  \frac{Q_A(x, y)}{\Tr(A) +2}}\Bigg)
\end{align}
This is the $\SU(2)$ Chern-Simons-Witten invariant of $M([A])$ at level $k$, see \cite[Equation (5.20)]{Jef}. Let $N(k):=8(k+2)$, the invariant $Z(M([A]), \SU(2), k)$ only depends on the conjugacy class of $[A] \bmod N(k)$ in $\SL(2, \ZZ/N(k))$. This follows from the duality between 3D Chern-Simons-Witten theory and 2D rational conformal field theory which constructs a representation  
\begin{align}\label{rep}
\Phi^{(k)}:\SL(2, \ZZ) \to \Aut(\C^{k+1})
\end{align}
which factors through $\SL(2, \ZZ/N(k))$ and has the property that 
\begin{align}
    Trace(\Phi^{(k)} (A)) = Z(M([A]), \SU(2), k)
\end{align}
For a fixed positive integer $k$ one can thus apply Sarnak's Chebotarev density theorem to obtain asymptotics of sums of \eqref{csw} over all hyperbolic conjugacy classes $[A]$ in $[\SL(2, \ZZ)]$. It would be interesting to also consider the asymptotics of these sums as $k\to \infty$.

\end{document}